\documentclass[12pt]{article}

\usepackage{algorithm}				
\usepackage[noend]{algpseudocode}	
\usepackage{amsmath}				
\usepackage{amssymb}				
\usepackage{amsthm}					
\usepackage{caption}
\usepackage{color}
\usepackage[margin=1in]{geometry}
\usepackage[hidelinks]{hyperref}
\usepackage{titling}

\usepackage{complexity}
\usepackage{graphicx}
\usepackage{tikz}
\usetikzlibrary{matrix, shapes}

\theoremstyle{plain}				
\newtheorem{theorem}{Theorem}
\newtheorem{lemma}[theorem]{Lemma}
\newtheorem{proposition}[theorem]{Proposition}
\newtheorem{corollary}[theorem]{Corollary}

\theoremstyle{definition}			
\newtheorem{definition}[theorem]{Definition}

\theoremstyle{remark}				

\thanksmarkseries{alph}

\newcommand*\ellipsed[1]{\tikz[baseline=(char.base)]{\node[shape=rounded rectangle,draw,inner sep=1.25pt] (char) {#1};}}
\newcommand*\dashedell[1]{\tikz[baseline=(char.base)]{\node[shape=rounded rectangle,draw,dashed,inner sep=1.25pt] (char) {#1};}}

\usepackage{pifont}                             
\newcommand{\cmark}{\ding{51}}
\newcommand{\xmark}{\ding{55}}

\newcommand{\from}{\colon}

\newcommand{\TDFA}{\textsf{2DFA-4W}}

\newcommand{\TDFATW}{\textsf{2DFA-3W}}
\newcommand{\TDFATWOS}{\textsf{2DFA-3W-1$\Sigma$}}
\newcommand{\TDFATWOW}{\textsf{2DFA-2W}}
\newcommand{\TDFATWOWOS}{\textsf{2DFA-2W-1$\Sigma$}}
\newcommand{\TNFA}{\textsf{2NFA-4W}}

\newcommand{\TNFATW}{\textsf{2NFA-3W}}
\newcommand{\TNFATWOS}{\textsf{2NFA-3W-1$\Sigma$}}
\newcommand{\TNFATWOW}{\textsf{2NFA-2W}}
\newcommand{\TNFATWOWOS}{\textsf{2NFA-2W-1$\Sigma$}}

\title{Complexity of Universality and Related Decision Problems for Unary Two-Dimensional Automata}
\author{
	Taylor J. Smith \thanks{Department of Computer Science, St.\ Francis Xavier University, Antigonish, Nova Scotia, Canada. Email: \href{mailto:tjsmith@stfx.ca}{\texttt{tjsmith@stfx.ca}}.}
}
\date{\today}

\begin{document}


\maketitle

\begin{abstract}
A two-dimensional automaton is able to move its input head through its input word in four directions: upward, downward, leftward, and rightward. If we prevent the input head from moving upward, then we obtain a three-way two-dimensional automaton; preventing both upward and leftward movements results in a two-way two-dimensional automaton. While much is known about the decidability and complexity properties of the two-dimensional automaton model, the unary variant of this model is less studied.

We show that the universality, equivalence, and inclusion problems for unary three-way deterministic two-dimensional automata are \co\NP-hard, while for the corresponding two-way model, the universality, equivalence, inclusion, and disjointness problems are in \P. We further show that the universality, equivalence, and inclusion problems for unary two-way nondeterministic two-dimensional automata are \co\NP-hard and in \ELEMENTARY; and the disjointness problem for the same model is \NL-hard and in \ELEMENTARY. Finally, we establish the decidability of a bounded variant of the universality problem for unary three-way nondeterministic two-dimensional automata, and show that this variant problem is \co\NP-complete.

\medskip


\medskip

\noindent\textit{MSC2020 classes:} 68Q45 (primary); 68Q17 (secondary).
\end{abstract}


\section{Introduction}\label{sec:introduction}

A two-dimensional automaton is a generalization of the classical one-dimensional finite automaton that takes as input arrays or matrices of symbols from an alphabet $\Sigma$. The two-dimensional automaton model was introduced by Blum and Hewitt~\cite{BlumHewitt19672DAutomata}.

The standard definition of a two-dimensional automaton, in which the input head is allowed to move in four directions through its input word, gives rise to a very powerful model of computation; indeed, two-dimensional automata are Turing-equivalent. Unfortunately, this computational power means that many common decision problems we can ask about this model are rendered undecidable. For this reason, we can restrict the model in certain ways: by preventing the input head of a two-dimensional automaton from moving upward or leftward. In the case where the input head cannot move upward, the model is called a three-way two-dimensional automaton; when the input head can move neither upward nor leftward, the model is called a two-way two-dimensional automaton. With these restricted variants of two-dimensional automata, we are able to regain many useful decidability properties.

We may also define a unary variant of the two-dimensional automaton model. When $\Sigma$ is a one-letter alphabet, the only interesting property of a two-dimensional input word is its dimension. In this case, we are able to regain even more decidability properties, although the landscape of decidability is still not yet settled for each unary model.

In this paper, we approach universality from two directions: we study the two-dimensional automaton model itself, which is a universal model of computation; and we study the decision problem of universality, which asks whether a given automaton accepts all input words over a particular alphabet, for unary restricted variants of two-dimensional automata. We also study the related decision problems of equivalence, inclusion, and disjointness for the same models. We prove various decidability and complexity results pertaining to unary two-dimensional automata: we show that the universality, equivalence, and inclusion problems are \co\NP-hard for the unary three-way deterministic model; the universality, equivalence, inclusion, and disjointness problems are in \P\ for the unary two-way deterministic model; the universality, equivalence, and inclusion problems for the unary two-way nondeterministic model are \coNP-hard and in \ELEMENTARY; and the disjointness problem for the same model is \NL-hard and in \ELEMENTARY. We also consider the long-standing problem of decidability of the unary three-way nondeterministic universality problem; while we do not resolve the question in full, we establish a positive result for a bounded variant of this problem and show that it is \co\NP-complete.


\section{Preliminaries}\label{sec:preliminaries}

In this section, we review basic definitions and results to be used in this paper. For further details about two-dimensional formal languages and automata, see the survey articles of Giammarresi and Restivo~\cite{GiammarresiRestivo19972DLanguages} and of Inoue and Takanami~\cite{Inoue19912DAutomataSurvey}, the book of Rosenfeld~\cite{Rosenfeld1979PictureLanguages}, or the past work of the author~\cite{Smith2019TwoDimensionalAutomata,Smith2021PhDThesis}.

A two-dimensional word is a finite array of cells, each labelled by a symbol from a finite alphabet $\Sigma$. If a two-dimensional word $W$ consists of $m \geq 1$ rows and $n \geq 1$ columns, then we say it has dimension $m \times n$, and we may index into row $i$ and column $j$ by writing $W[i,j]$. By convention, we take $W[0, 0]$ to be the top-left corner of $W$ and we take $W[m-1, n-1]$ to be the bottom-right corner of $W$. For a unary alphabet $\Sigma = \{\texttt{a}\}$, we represent a two-dimensional word having dimension $m \times n$ as $\texttt{a}^{m \times n}$. We denote the special language of all two-dimensional words over $\Sigma$ having $m \geq 1$ rows and $n \geq 1$ columns by the notation $\Sigma^{++}$.

When a two-dimensional word is given as input to a model of computation (i.e., written to an input tape), we label the cells around the two-dimensional word by a special boundary marker $\texttt{\#} \not\in \Sigma$. Perhaps the simplest model of computation used to process two-dimensional words is the two-dimensional automaton, which has a finite-state control and is capable of moving its input head in four directions through its input: up, down, left, and right (denoted $U$, $D$, $L$, and $R$, respectively). Without loss of generality, we may assume that the two-dimensional automaton begins its computation on an input word $W$ with its input head positioned at $W[0, 0]$. If the input head encounters a boundary marker \texttt{\#}, then we may assume it automatically returns to the input word from the border by undoing whatever previous move it made. We assume that a two-dimensional automaton accepts by entering a designated accepting state, $q_{\text{accept}}$, and that it halts and accepts immediately upon entering $q_{\text{accept}}$. Taking all of this together gives us the formal definition of our chosen model of computation.

\begin{definition}[Deterministic two-dimensional automaton]\label{def:2DFA4W}
A deterministic two-dimensional automaton (\TDFA) is a tuple $(Q, \Sigma, \delta, q_{0}, q_{\text{accept}})$, where $Q$ is a finite set of states, $\Sigma$ is the input alphabet (with $\texttt{\#} \not\in \Sigma$ acting as a boundary marker), $\delta \from (Q \setminus \{q_{\text{accept}}\}) \times (\Sigma \cup \{\texttt{\#}\}) \to Q \times \{U, D, L, R\}$ is the partial transition function, and $q_{0}, q_{\text{accept}} \in Q$ are the initial and accepting states, respectively.
\end{definition}

We can modify a two-dimensional automaton to be nondeterministic (\TNFA) in the usual way by changing the transition function to map to the power set $2^{Q \times \{U, D, L, R\}}$. A two-dimensional automaton $\mathcal{A}$ is said to be unary if $|\Sigma| = 1$. In this case, the only property of an input word the automaton is capable of recognizing is its dimension $m \times n$, so we can identify a unary two-dimensional language $L(\mathcal{A})$ by way of its accepted dimension set $D(\mathcal{A}) = \{(m,n) \in \mathbb{N}_{>0} \times \mathbb{N}_{>0} \mid \texttt{a}^{m \times n} \in L(\mathcal{A})\}$, where the notation $\mathbb{N}_{>0}$ denotes the set $\mathbb{N} \setminus \{0\}$. If a two-dimensional automaton is unary, we append to its notation the suffix \textsf{-1$\Sigma$}.

The suffix \textsf{-4W} in the previous notation indicates that our two-dimensional automata are ``four-way" automata; that is, they may move their input heads in any of the four directions $\{U, D, L, R\}$. By restricting the movement of the input head, we obtain a number of restricted variants of the two-dimensional automaton model.

\begin{definition}[Three-way two-dimensional automaton]\label{def:2DFA3W}
A three-way two-dimensional automaton (\TDFATW/\TNFATW) is a tuple $(Q, \Sigma, \delta, q_{0}, q_{\text{accept}})$ as in Definition~\ref{def:2DFA4W}, where the transition function $\delta$ is restricted to use only the directions $\{D, L, R\}$.
\end{definition}

\begin{definition}[Two-way two-dimensional automaton]\label{def:2DFA2W}
A two-way two-dimensional automaton (\TDFATWOW/\TNFATWOW) is a tuple $(Q, \Sigma, \delta, q_{0}, q_{\text{accept}})$ as in Definition~\ref{def:2DFA4W}, where the transition function $\delta$ is restricted to use only the directions $\{D, R\}$.
\end{definition}

Both the three-way and two-way two-dimensional automaton variants can be either deterministic or nondeterministic, depending on their transition function $\delta$. Since its input head movements are restricted, if a three-way two-dimensional automaton reads a boundary marker after making a downward move, it can read only boundary markers for the remainder of its computation; the same applies if a two-way two-dimensional automaton reads a boundary marker after making either a downward or a rightward move.

It is known that (four-way) two-dimensional automata can, in some sense, simulate one-dimensional Turing machines~\cite{BlumSakoda1977CapabilityFiniteAutomata}; this is achieved by modelling a two-dimensional automaton as a two-counter machine and using the known result that two-counter machines are Turing-equivalent~\cite{Minsky1967Computation}. Additionally, unlike with typical (one-dimensional) finite automata, deterministic two-dimensional automata are strictly weaker than nondeterministic two-dimensional automata~\cite{BlumHewitt19672DAutomata}. As one might expect, introducing greater restrictions limits the recognition power of a two-dimensional automaton, so that the three-way model is strictly weaker than the four-way model~\cite{Rosenfeld1979PictureLanguages} and, likewise, the two-way model is strictly weaker than the three-way model.

This paper considers common decision problems one may ask about a two-dimensional automaton $\mathcal{A}$ and its language $L(\mathcal{A})$. In addition to the membership problem (given an automaton $\mathcal{A}$ and an input word $W$, is $W \in L(\mathcal{A})$?) and the emptiness problem (is $L(\mathcal{A}) = \emptyset$?), we focus on the decision problems of universality (is $L(\mathcal{A}) = \Sigma^{++}$?), equivalence (is $L(\mathcal{A}) = L(\mathcal{B})$ for some other given automaton $\mathcal{B}$?), inclusion (is $L(\mathcal{A}) \subseteq L(\mathcal{B})$?), and disjointness (is $L(\mathcal{A}) \cap L(\mathcal{B}) = \emptyset$?). For unary two-dimensional automata, we may translate each of the definitions of these decision problems to use accepted dimension sets instead of languages without too much effort.


\begin{table}
\centering
\begin{tabular}{c | p{2.75cm} p{2.75cm} p{2.75cm} p{2.75cm}}
				& \TDFATW	& \TNFATW	& \TDFATWOW	& \TNFATWOW \\
\hline
membership		& \cmark	& \cmark	& \cmark	& \cmark \\
emptiness		& \cmark	& \cmark	& \cmark	& \cmark \\
universality	& \cmark	& \xmark	& \cmark	& \xmark \\
equivalence		& \textbf{?}& \xmark	& \cmark	& \xmark \\
inclusion		& \xmark	& \xmark	& \cmark	& \xmark \\
disjointness	& \xmark	& \xmark	& \cmark	& \textbf{?} \\
\end{tabular}
\caption{Decidability results for restricted two-dimensional automaton models.}
\label{tab:2Ddecidability}
\end{table}


\begin{table}
\centering
\begin{tabular}{c | p{2.75cm} p{2.75cm} p{2.75cm} p{2.75cm}}
				& \TDFATWOS	& \TNFATWOS		& \TDFATWOWOS	& \TNFATWOWOS \\
\hline
membership		& \cmark	& \cmark		& \cmark		& \cmark \\
emptiness		& \cmark	& \cmark		& \cmark		& \cmark \\
universality	& \cmark$_{\,\scriptsize{\ellipsed{\co\NP\text{-hard}}}}$	& \dashedell{$\ast\,\,^{\scriptsize{\co\NP}}_{\scriptsize{\co\NP\text{-hard}}}$}	& \cmark$^{\,\scriptsize{\ellipsed{\P}}}$		& \ellipsed{\cmark$\,^{\scriptsize{\ELEMENTARY}}_{\scriptsize{\co\NP\text{-hard}}}$} \\
equivalence		& \cmark$_{\,\scriptsize{\ellipsed{\co\NP\text{-hard}}}}$	& \textbf{?}	& \cmark$^{\,\scriptsize{\ellipsed{\P}}}$		& \ellipsed{\cmark$\,^{\scriptsize{\ELEMENTARY}}_{\scriptsize{\co\NP\text{-hard}}}$} \\
inclusion		& \cmark$_{\,\scriptsize{\ellipsed{\co\NP\text{-hard}}}}$	& \textbf{?}	& \cmark$^{\,\scriptsize{\ellipsed{\P}}}$		& \ellipsed{\cmark$\,^{\scriptsize{\ELEMENTARY}}_{\scriptsize{\co\NP\text{-hard}}}$} \\
disjointness	& \cmark	& \textbf{?}	& \cmark$^{\,\scriptsize{\ellipsed{\P}}}$		& \ellipsed{\cmark$\,^{\scriptsize{\ELEMENTARY}}_{\scriptsize{\NL\text{-hard}}}$} \\
\end{tabular}
\caption{Decidability results for unary restricted two-dimensional automaton models. New results presented in this paper are circled. A partial result presented in this paper is dash-circled. Superscripts indicate upper bounds. Subscripts indicate lower bounds.}
\label{tab:2Dunarydecidability}
\end{table}

In terms of decidability properties, (four-way) two-dimensional automata present nothing of particular interest; indeed, only the membership problem is decidable for this model. Thus, here we focus on the restricted variants of two-dimensional automata, which admit some more positive decidability results. Known decidability results are summarized in Table~\ref{tab:2Ddecidability} for general alphabets and in Table~\ref{tab:2Dunarydecidability} for unary alphabets, where \cmark\ denotes a decidable problem and \xmark\ denotes an undecidable problem. (The asterisk $\ast$ denotes a positive decidability result for a special bounded case discussed in Section~\ref{sec:3Wnondetbounded}; the general entry remains \textbf{?}.)

Of the decision problems known to be decidable for restricted two-dimensional automaton models, the complexity of deciding these problems has been studied in part. Since testing membership is essentially a reachability problem, the deterministic membership problem is \L-complete while the nondeterministic version is \NL-complete~\cite{Lindgren1998Complexity2DPatterns}. The \TNFATWOS\ emptiness problem is \NP-complete~\cite{SmithSalomaa20212D3WProperties}, and an upper bound of \NP\ follows for the \TDFATWOS\ emptiness problem. In the two-way case, the \TNFATWOW\ emptiness problem is in \P~\cite{SmithSalomaa20212D3WProperties}, and so this same bound transfers over to the \TNFATWOWOS\ and \TDFATWOWOS\ emptiness problems.


\section{Unary Three-Way Deterministic Lower Bounds}\label{sec:3Wlowerbound}

As Table~\ref{tab:2Dunarydecidability} observes, the problems of universality, equivalence, and inclusion are all known to be decidable for unary three-way deterministic two-dimensional automata. The decidability of the universality problem was established by Inoue and Takanami~\cite{InoueTakanami1980DecisionProblems2DAutomata}. The decidability of both the equivalence and the inclusion problems was established by Kinber~\cite{Kinber1985ThreeWayAutomataOneLetter}. Here, we consider the complexity of each of these decision problems for this model.

It is well known that Karp~\cite{Karp1972Reducibility} established the \NP-completeness of \textsc{3-CNF-Sat}: given a Boolean formula $\phi$ in conjunctive normal form, where each clause contains at most three literals, determine whether $\phi$ is satisfiable. As a consequence, the complementary problem \textsc{3-CNF-Unsat} is \co\NP-complete. We will use this fact to establish the following theorem giving a lower bound on the complexity of the \TDFATWOS\ universality problem.

\begin{theorem}\label{thm:unary2DFA3Wuniversality}
The universality problem for three-way deterministic two-dimensional automata over a unary alphabet is \co\NP-hard.

\begin{proof}
We prove this theorem via a reduction from \textsc{3-CNF-Unsat}. Let $\phi = C_{1} \wedge \dots \wedge C_{s}$ be a Boolean formula in conjunctive normal form with at most three literals per clause, and suppose the variables appearing in $\phi$ are $x_{1}, \dots, x_{t}$. Enumerate the first $t$ prime numbers $p_{1}, \dots, p_{t}$, and associate to each variable $x_{i}$ the $i$th prime number $p_{i}$ in this enumeration. By the prime number theorem, we know that $p_{t} \sim t\log(t)$, so $p_{i} \leq p_{t} \in O(t\log(t))$ for all $i \leq t$; in other terms, all primes in the enumeration are polynomially bounded in $t$.

We construct a unary three-way deterministic two-dimensional automaton $\mathcal{A}_{\phi}$. Given an input word $W$ of dimension $m \times n$, we interpret the number of columns $n$ as an encoding of an assignment of truth values to variables:
\begin{equation*}
x_{i} = 
\begin{cases}
\text{false},	& \text{if } n \equiv 0 \bmod p_{i}; \\
\text{true},	& \text{if } n \equiv 1 \bmod p_{i}.
\end{cases}
\end{equation*}
Let $(\lambda_{1}, \dots, \lambda_{\ell})$ be a list of occurrences of literals in $\phi$ listed clause-by-clause, and record in the finite-state control of $\mathcal{A}_{\phi}$ which literals end each clause. By our truth value assignment, a positive literal $x_{i}$ is false when $n \equiv 0 \bmod p_{i}$, while a negative literal $\neg x_{i}$ is false when $n \equiv 1 \bmod p_{i}$.

The automaton $\mathcal{A}_{\phi}$ begins by computing the residue for each literal occurrence on the fly. If the current literal $\lambda_{c}$ uses variable $x_{i}$, then on the row of the input word corresponding to $\lambda_{c}$, $\mathcal{A}_{\phi}$ moves from one boundary to the other while cycling through states $r_{0}^{i}, r_{1}^{i}, \dots, r_{p_{i} - 1}^{i}$. Each time the input head makes a leftward or rightward move and reads the unary alphabet symbol \texttt{a}, $\mathcal{A}_{\phi}$ transitions from state $r_{j}^{i}$ to state $r_{j + 1 \bmod p_{i}}^{i}$. Upon reaching the end of the row, $\mathcal{A}_{\phi}$ knows the value of $n \bmod p_{i}$; if $n \bmod p_{i} \not\in \{0, 1\}$, then $n$ does not encode an assignment of truth values and $\mathcal{A}_{\phi}$ defaults to accepting.

Then, $\mathcal{A}_{\phi}$ uses the sign of $\lambda_{c}$ stored in its finite-state control to determine whether $\lambda_{c}$ is true under the encoded assignment. The automaton stores in its finite-state control a flag value that is initialized to be false at the beginning of each clause, indicating whether some literal checked so far in that clause is true and updating if necessary after reading each row. If $\lambda_{c}$ is the last literal in the current clause and the flag remains false, then that clause is falsified and $\mathcal{A}_{\phi}$ accepts. Otherwise, if $\lambda_{c}$ is the last literal in the current clause and the flag is true, then $\mathcal{A}_{\phi}$ resets the flag to be false and proceeds to the next clause. If all clauses are satisfied, then $\mathcal{A}_{\phi}$ enters a dead state and does not accept.

Altogether, $\mathcal{A}_{\phi}$ requires one row in its input word for each literal occurrence. If the input word consists of fewer than $T = \ell$ rows, then $\mathcal{A}_{\phi}$ will reach the bottom boundary prematurely and default to accepting.

We now prove correctness.
\begin{itemize}
\item Suppose $\phi$ is an unsatisfiable formula, and let $W$ be a unary two-dimensional input word of dimension $m \times n$. If $m < T$, then $\mathcal{A}_{\phi}$ accepts. Likewise, if $m \geq T$ and $n$ is not a valid assignment encoding, then $\mathcal{A}_{\phi}$ accepts. Otherwise, $m \geq T$ and $n$ is a valid assignment encoding. Since $\phi$ is unsatisfiable, this assignment encoding falsifies at least one clause, and when $\mathcal{A}_{\phi}$ checks the first such falsified clause, it accepts. Therefore, every unary two-dimensional input word is accepted, and $L(\mathcal{A}_{\phi})$ is universal.

\item Suppose $\phi$ is a satisfiable formula, and let $\sigma$ denote a satisfying assignment of truth values to variables. For each $i$, let
\begin{equation*}
b_{i} = 
\begin{cases}
0,		& \text{if } \sigma(x_{i}) = \text{false}; \\
1,		& \text{if } \sigma(x_{i}) = \text{true}.
\end{cases}
\end{equation*}
Since each of $p_{1}, \dots, p_{t}$ is pairwise coprime, by the Chinese remainder theorem, we know that there exists a positive integer $n$ such that $n \equiv b_{i} \bmod p_{i}$ for each $i$. Choose any $m \geq T$. Then, in the corresponding unary two-dimensional input word of dimension $m \times n$, $\mathcal{A}_{\phi}$ will never accept. Thus, $L(\mathcal{A}_{\phi})$ is not universal.
\end{itemize}
Since \textsc{3-CNF-Unsat} is \co\NP-complete and we can construct $\mathcal{A}_{\phi}$ using $O(\ell \cdot t\log(t))$ states in time polynomial in $|\phi|$, testing universality is \co\NP-hard.
\end{proof}
\end{theorem}

From Theorem~\ref{thm:unary2DFA3Wuniversality}, we obtain similar lower bounds for the \TDFATWOS\ equivalence and inclusion problems.

\begin{corollary}\label{cor:unary2DFA3Wequivalenceinclusion}
The equivalence and inclusion problems for three-way deterministic two-dimensional automata over a unary alphabet are \co\NP-hard.

\begin{proof}
Let $\Sigma$ be a unary alphabet. We can construct a unary three-way deterministic two-dimensional automaton $\mathcal{U}$ recognizing $\Sigma^{++}$ in a straightforward manner: after reading the first input symbol, $\mathcal{U}$ transitions immediately to $q_{\text{accept}}$. By Theorem~\ref{thm:unary2DFA3Wuniversality}, we know that testing universality is \co\NP-hard.

Given a unary three-way deterministic two-dimensional automaton $\mathcal{A}$, construct a pair $(\mathcal{A}, \mathcal{U})$. Then deciding whether $L(\mathcal{A})$ is universal is at least as hard as deciding whether $L(\mathcal{A}) = L(\mathcal{U})$. Since $\mathcal{U}$ is fixed and constructing the pair $(\mathcal{A}, \mathcal{U})$ can be done in polynomial time, testing equivalence is \co\NP-hard.

Similarly, given $\mathcal{A}$ as before, construct a pair $(\mathcal{U}, \mathcal{A})$. Then deciding whether $L(\mathcal{A})$ is universal is at least as hard as deciding whether $L(\mathcal{U}) \subseteq L(\mathcal{A})$. Again, since $\mathcal{U}$ is fixed and constructing the pair $(\mathcal{U}, \mathcal{A})$ can be done in polynomial time, testing inclusion is \co\NP-hard.
\end{proof}
\end{corollary}

Note that we cannot apply the same logic to argue that testing \TDFATWOS\ emptiness or disjointness is \co\NP-hard. While this model is closed under complement, meaning we may test emptiness by way of testing complement-universality, there is no known polynomial-time construction to produce a polynomial-size unary three-way deterministic two-dimensional automaton for $\overline{L(\mathcal{A})}$, and testing disjointness is fundamentally a question of testing emptiness instead of universality.


\section[Unary Three-Way Nondeterministic Results for Bounded-Size Inputs]{Unary Three-Way Nondeterministic Results for\\Bounded-Size Inputs}\label{sec:3Wnondetbounded}

We now move on to considering the universality problem in the unary three-way nondeterministic case. For general alphabets, it is known that universality is undecidable~\cite{InoueTakanami1980DecisionProblems2DAutomata}, but the question remains unresolved for unary alphabets. Here, we do not attempt to resolve the question entirely, but instead we obtain a partial positive result by establishing decidability for a bounded variant of the universality problem.

For finite sets $R, C \subseteq \mathbb{N}_{>0}$, let $D_{R, C} = (R \times \mathbb{N}_{>0}) \cup (\mathbb{N}_{>0} \times C)$ denote the domain of all two-dimensional words having dimension $r \times \mathbb{N}_{>0}$ for $r \in R$ or having dimension $\mathbb{N}_{>0} \times c$ for $c \in C$.

As the first step toward obtaining our partial positive decidability result, we must develop two reductions depending on which dimension we choose to fix: the number of columns or the number of rows.

\paragraph{Fixed columns.} Fix a value $c \geq 1$. Define the fixed-column-universality problem to be as follows: given a unary three-way nondeterministic two-dimensional automaton $\mathcal{A}$ over the alphabet $\Sigma = \{\texttt{a}\}$, does $\texttt{a}^{m \times c} \in L(\mathcal{A})$ for all $m \geq 1$?

The following lemma allows us to reduce this problem to an equivalent problem framed in terms of a unary nondeterministic finite automaton (in one dimension), which will then enable us to use well-known properties of this model in establishing the main theorem of the section.

\begin{lemma}\label{lem:fixedcolumn}
For every unary three-way nondeterministic two-dimensional automaton $\mathcal{A}$ and for all $c \geq 1$, we can construct a unary nondeterministic finite automaton $\mathcal{B}_{\mathcal{A},c}$ over the alphabet $\Gamma = \{\texttt{b}\}$ such that, for all $m \geq 1$, $\texttt{b}^{m} \in L(\mathcal{B}_{\mathcal{A},c})$ if and only if $\texttt{a}^{m \times c} \in L(\mathcal{A})$.

\begin{proof}
Take a row-entry configuration to be a pair $(q, j)$, where $q \neq q_{\text{accept}}$ is a state of $\mathcal{A}$ and $j \in \{0, 1, \dots, c + 1\}$ is the column in which the input head enters that particular row in state $q$. Column numbers $0$ and $c + 1$ represent the left and right boundary columns, respectively.

Observe that, since $c$ is fixed, there are finitely many row-entry configurations. Thus, we may build a graph describing all input head movement within a given row. The vertices of this graph are labelled by pairs of states of $\mathcal{A}$ and a column number from the set $\{0, 1, \dots, c + 1\}$. Edges between vertices indicate leftward and rightward moves within the row.

To indicate downward moves from one row to another, define a relation $\rightsquigarrow_{c}$ on row-entry configurations in the following way: $(q, j) \rightsquigarrow_{c} (r, j')$ if, starting in a state $q$ and column $j$ of a row with $c$ columns, $\mathcal{A}$ makes a finite number of leftward and rightward moves within that row before making its first downward move into the next row at column $j'$, where it transitions to state $r$. Since each row-entry configuration graph is finite, we may compute the relation $\rightsquigarrow_{c}$ using standard reachability techniques.

Additionally, define two finite sets: $F_{c}$, consisting of all row-entry configurations from which $\mathcal{A}$ can accept before making a downward move; and $G_{c}$, consisting of all row-entry configurations from which $\mathcal{A}$ can accept if the current row is the last row (i.e., if the next downward move reaches the bottom boundary of the input word).

We can now construct the unary nondeterministic finite automaton $\mathcal{B}_{\mathcal{A},c}$. Each state in the state set of $\mathcal{B}_{\mathcal{A},c}$ is labelled by a row-entry configuration. The initial state corresponds to the row-entry configuration $(q_{0}, 1)$. We also add one accepting sink state $q_{\text{s}}$ that loops on \texttt{b}, and one standard accepting state $q_{\text{a}}$ with no outgoing transitions. For each $q \rightsquigarrow_{c} r$, add a transition $r \in \delta(q, \texttt{b})$. If $q \in F_{c}$, then also add a transition $q_{\text{s}} \in \delta(q, \texttt{b})$. Lastly, if $q \in G_{c}$, then also add a transition $q_{\text{a}} \in \delta(q, \texttt{b})$.

After reading $i$ input symbols without entering either $q_{\text{s}}$ or $q_{\text{a}}$, the reachable states of $\mathcal{B}_{\mathcal{A},c}$ correspond to the row-entry configurations in which $\mathcal{A}$ can enter row $i + 1$. Therefore, for each $m \geq 1$, $\mathcal{B}_{\mathcal{A},c}$ accepts the input word $\texttt{b}^{m}$ if and only if $\mathcal{A}$ accepts the input word $\texttt{a}^{m \times c}$.
\end{proof}
\end{lemma}

\paragraph{Fixed rows.} Fix a value $r \geq 1$. Define the fixed-row-universality problem to be as follows: given a unary three-way nondeterministic two-dimensional automaton $\mathcal{A}$ over the alphabet $\Sigma = \{\texttt{a}\}$, does $\texttt{a}^{r \times n} \in L(\mathcal{A})$ for all $n \geq 1$?

Again, the following lemma reduces this problem to an equivalent problem for another, more well-known computational model. While Lemma~\ref{lem:fixedcolumn} made use of a (one-way) nondeterministic finite automaton to traverse two-dimensional rows in the downward direction, here we must use a two-way nondeterministic finite automaton to move leftward and rightward between two-dimensional columns.

\begin{lemma}\label{lem:fixedrow}
For every unary three-way nondeterministic two-dimensional automaton $\mathcal{A}$ and for all $r \geq 1$, we can construct a unary two-way nondeterministic finite automaton $\mathcal{C}_{\mathcal{A},r}$ over the alphabet $\Gamma = \{\texttt{c}\}$ such that, for all $n \geq 1$, $\texttt{c}^{n} \in L(\mathcal{C}_{\mathcal{A},r})$ if and only if $\texttt{a}^{r \times n} \in L(\mathcal{A})$.

\begin{proof}
We construct the unary two-way nondeterministic finite automaton $\mathcal{C}_{\mathcal{A},r}$ in the following way. Given an input word of length $n$, the position of the input head of $\mathcal{C}_{\mathcal{A},r}$ represents the current column being read by $\mathcal{A}$. If the input head of $\mathcal{C}_{\mathcal{A},r}$ reads either the left or the right endmarker of its input word, then this is interpreted as $\mathcal{A}$ reading a boundary marker \texttt{\#}. The finite-state control of $\mathcal{C}_{\mathcal{A},r}$ keeps track of both the current state of $\mathcal{A}$ and a row number from the set $\{1, \dots, r, \bot\}$, where $\bot$ represents the bottom boundary row. Again, if the row value is $\bot$, then this is interpreted as $\mathcal{A}$ reading a boundary marker \texttt{\#}.

If the input head of $\mathcal{A}$ moves leftward or rightward, then $\mathcal{C}_{\mathcal{A},r}$ simulates this by moving its input head in the same direction. If $\mathcal{A}$ moves downward, then $\mathcal{C}_{\mathcal{A},r}$ does not move its input head, but does increment the row value stored in its finite-state control if it is less than $r$ or updates it to $\bot$ if it is equal to $r$.

Lastly, if $\mathcal{A}$ enters its accepting state $q_{\text{accept}}$, then $\mathcal{C}_{\mathcal{A},r}$ likewise accepts. Since the simulation performed by $\mathcal{C}_{\mathcal{A},r}$ simulates the states and input word indices visited by $\mathcal{A}$, $\mathcal{C}_{\mathcal{A},r}$ accepts the input word $\texttt{c}^{n}$ if and only if $\mathcal{A}$ accepts the input word $\texttt{a}^{r \times n}$.
\end{proof}
\end{lemma}

One may now be able to recognize that, by combining multiple instances of the fixed-column-universality problem (for various values of $c$) and the fixed-row-universality problem (for various values of $r$), we can obtain the aforementioned bounded variant of the universality problem for the unary three-way nondeterministic two-dimensional automaton model. We will refer to this as the $D_{R, C}$-universality problem, following the notation we introduced at the beginning of this section.

It is known that the universality problem for unary nondeterministic finite automata\footnote{Note that, in Section~\ref{sec:preliminaries}, we assumed each two-dimensional word contains at least one row and one column, while in one dimension, a universal finite automaton $\mathcal{M}$ accepts the empty word $\epsilon$. We account for this edge case by applying a preprocessing step that simulates $\mathcal{M}$ on nonempty input words and handles $\epsilon$ separately.} is \co\NP-complete; this result is due to Stockmeyer and Meyer~\cite{StockmeyerMeyer1973WordProblems}. Likewise, the universality problem for unary two-way nondeterministic finite automata is \co\NP-complete. Galil~\cite{Galil1976HierarchiesCompleteProblems} proved that the nonemptiness problem for unary two-way nondeterministic finite automata is in \NP, and since testing nonemptiness is equivalent to testing nonuniversality of the complement language, we can use a polynomial-size construction of Geffert, Mereghetti, and Pighizzini~\cite{Geffert2007ComplementingTwoWay} to obtain the complement unary language; hardness follows by the aforementioned one-way nondeterministic finite automaton result. Combining these facts with Lemma~\ref{lem:fixedcolumn} and Lemma~\ref{lem:fixedrow}, we obtain the following theorem.

\begin{theorem}\label{thm:DRCuniversalitycoNPcomplete}
The $D_{R, C}$-universality problem for three-way nondeterministic two-dimensional automata over a unary alphabet is \co\NP-complete for every fixed finite $R, C \subseteq \mathbb{N}_{>0}$, or when finite sets $R, C \subseteq \mathbb{N}_{>0}$ are given in unary as part of the input.

\begin{proof}
Let $\mathcal{A}$ be a unary three-way nondeterministic two-dimensional automaton. For each $c \in C$, construct a unary nondeterministic finite automaton $\mathcal{B}_{\mathcal{A}, c}$ as in the proof of Lemma~\ref{lem:fixedcolumn}. Likewise, for each $r \in R$, construct a unary two-way nondeterministic finite automaton $\mathcal{C}_{\mathcal{A}, r}$ as in the proof of Lemma~\ref{lem:fixedrow}. If $R$ and $C$ are provided in unary encoding, then we can construct this family of automata in time polynomial in the length of the input. Moreover, $L(\mathcal{A})$ is $D_{R, C}$-universal if and only if $\mathcal{B}_{\mathcal{A}, c}$ is universal for all $c \in C$ and $\mathcal{C}_{\mathcal{A}, r}$ is universal for all $r \in R$.

Since the universality problem for unary nondeterministic finite automata (as well as that of the two-way variant) is in \co\NP, the nonuniversality of either $\mathcal{B}_{\mathcal{A}, c}$ or $\mathcal{C}_{\mathcal{A}, r}$ is witnessed by a polynomial-size \NP\ certificate. Thus, the $D_{R, C}$-nonuniversality problem is in \NP, and so the $D_{R, C}$-universality problem is in \co\NP.

We establish hardness via a reduction from the universality problem for unary nondeterministic finite automata. Let $\mathcal{M}$ be a unary nondeterministic finite automaton. We construct a unary three-way nondeterministic two-dimensional automaton $\mathcal{A}_{\mathcal{M}}$ with the property that $L(\mathcal{A}_{\mathcal{M}})$ is $D_{R, C}$-universal if and only if $L(\mathcal{M})$ is universal.
\begin{itemize}
\item Suppose $R \neq \emptyset$, and choose some $r_{0} \in R$. Given an arbitrary input word $\texttt{a}^{m \times n}$, $\mathcal{A}_{\mathcal{M}}$ ignores the number of rows $m$ and simulates the computation of $\mathcal{M}$ on the input word $\texttt{a}^{n}$. Then, $\mathcal{A}_{\mathcal{M}}$ accepts $\texttt{a}^{m \times n}$ if and only if $\mathcal{M}$ accepts $\texttt{a}^{n}$. If $L(\mathcal{M})$ is not universal, then there exists some $n_{0} \geq 1$ such that $\texttt{a}^{n_{0}} \not\in L(\mathcal{M})$; this means that $\mathcal{A}_{\mathcal{M}}$ rejects the input word $\texttt{a}^{r_{0} \times n_{0}}$, and so $L(\mathcal{A}_{\mathcal{M}})$ is not $D_{R, C}$-universal.

\item Suppose $R = \emptyset$ but $C \neq \emptyset$, and choose some $c_{0} \in C$. Given an arbitrary input word $\texttt{a}^{m \times n}$, $\mathcal{A}_{\mathcal{M}}$ first checks whether $n \neq c_{0}$ and, if this is the case, $\mathcal{A}_{\mathcal{M}}$ accepts. Otherwise, if $n = c_{0}$, then $\mathcal{A}_{\mathcal{M}}$ simulates the computation of $\mathcal{M}$ on the input word $\texttt{a}^{m}$ and accepts at the bottom boundary if and only if $\mathcal{M}$ accepts. If $L(\mathcal{M})$ is not universal, then there exists some $m_{0} \geq 1$ such that $\texttt{a}^{m_{0}} \not\in L(\mathcal{M})$; this means that $\mathcal{A}_{\mathcal{M}}$ rejects the input word $\texttt{a}^{m_{0} \times c_{0}}$, and so $L(\mathcal{A}_{\mathcal{M}})$ is not $D_{R, C}$-universal.
\end{itemize}
Since the universality problem for unary nondeterministic finite automata is \co\NP-complete and we can reduce from this problem to the $D_{R, C}$-universality problem, testing $D_{R, C}$-universality is \co\NP-hard.
\end{proof}
\end{theorem}


\section{Unary Two-Way Results}\label{sec:2W}

When considering the two-way two-dimensional automaton model, it will be beneficial for us to recall a definition from the literature that will allow us to simplify our proofs slightly.

\begin{definition}[IBR-accepting two-way two-dimensional automaton~\cite{SmithSalomaa20232DProjection}]\label{def:IBR}
An IBR-accepting two-way two-dimensional automaton $\mathcal{A}$ is a tuple $(Q, \Sigma, \delta, q_{0}, q_{\text{accept}})$ as in Definition~\ref{def:2DFA2W} where, when the input head reads a boundary marker \texttt{\#}, $\mathcal{A}$ either enters $q_{\text{accept}}$ in the next transition or the transition is undefined.
\end{definition}

The abbreviation ``IBR-accepting" refers to the automaton ``immediately bottom-right accepting"; by this, we mean that the automaton immediately halts and accepts if, upon reading a boundary marker \texttt{\#}, the state $q_{\text{accept}}$ is reachable from its current state. Every two-way two-dimensional automaton can be converted to an equivalent IBR-accepting automaton of the same size, so throughout this section, when we refer to a two-way two-dimensional automaton, we will assume it is IBR-accepting.


\subsection{Deterministic Complexity}\label{subsec:2DFA2W}

One may have noted in Table~\ref{tab:2Ddecidability} that every common decision problem is decidable for two-way deterministic two-dimensional automata over a general alphabet. The \TDFATWOW\ emptiness problem can be decided by searching for an input word dimension yielding an accepting computation, and since the model is closed under complement, we obtain a decision procedure for the \TDFATWOW\ universality problem as well. Equivalence and inclusion were proved decidable by the author and Salomaa~\cite{SmithSalomaa20212D3WProperties}, and the author later extended decidability to the disjointness problem~\cite{Smith2021PhDThesis}. Each of these positive decidability results naturally transfers over to the \TDFATWOWOS\ problems.

Here, we extend these results to show that the unary variants of each of these problems can be decided efficiently. This efficiency is achieved by virtue of the simplified behaviour of the transition function over a unary alphabet allowing us to model each decision problem as an integer linear program.

\begin{theorem}\label{thm:2DFA2WP}
The universality, equivalence, inclusion, and disjointness problems for two-way deterministic two-dimensional automata over a unary alphabet are in \P.

\begin{proof}
Let $\mathcal{A}$ be a unary two-way deterministic two-dimensional automaton. Since $\mathcal{A}$ is defined over $\Sigma = \{\texttt{a}\}$, consider the sequence of states $q_{0}, q_{1}, q_{2}, \dots$ that arises by repeatedly reading the symbol \texttt{a}; that is, the sequence produced by computing $\delta(q_{i}, \texttt{a}) = (q_{i+1}, X_{i})$, where $X_{i} \in \{D, R\}$, until either $q_{\text{accept}}$ is reached, a nonaccepting state with no outgoing transition on \texttt{a} is reached, or a nonaccepting state repeats in the sequence. The sequence of $\{D, R\}$ moves is monotone, and since $\mathcal{A}$ is deterministic, the sequence of states is unique.

For each transition $i \geq 0$, let $d_{i} = |\{j < i \mid X_{j} = D\}|$ and $r_{i} = |\{j < i \mid X_{j} = R\}|$ denote the number of downward and rightward moves, respectively, made before transition $i$. Immediately before $\mathcal{A}$ follows transition $i$, its input head is at row $1 + d_{i}$ and column $1 + r_{i}$ in its input word.

We classify the sequence of states visited by $\mathcal{A}$ into one of the following cases and, for each input word of dimension $m \times n$, we place its dimension $(m, n)$ into an accepting set $A_{\mathcal{A}}$ or a nonaccepting set $R_{\mathcal{A}}$ as follows:
\begin{itemize}
\item Suppose the sequence of states eventually reaches $q_{\text{accept}}$. Let $t$ be the least index such that $q_{t} = q_{\text{accept}}$ in the sequence. For each $j < t - 1$, we record the input word dimensions $(m, n)$ for which transition $j$ would move the input head of $\mathcal{A}$ into the bottom or right boundary before reaching $q_{\text{accept}}$.
	\begin{itemize}
	\item If $X_{j} = D$, then this happens when the input word dimensions satisfy $m = 1 + d_{j}$ and $n \geq 1 + r_{j}$.
	
	\item If $X_{j} = R$, then this happens when the input word dimensions satisfy $m \geq 1 + d_{j}$ and $n = 1 + r_{j}$.
	\end{itemize}
We categorize each case as accepting or rejecting based on the IBR-accepting behaviour of $\mathcal{A}$ upon reaching the boundary. If $t = 0$, then $\mathcal{A}$ accepts every input word. Otherwise, all remaining dimensions satisfy $m \geq 1 + d_{t - 1}$ and $n \geq 1 + r_{t - 1}$; for input words having these dimensions, $q_{\text{accept}}$ is reached before the input head encounters a boundary, so $\mathcal{A}$ accepts and we categorize that dimension into $A_{\mathcal{A}}$.

\item Suppose the sequence of states reaches a nonaccepting state with no outgoing transition on \texttt{a}. Let $t$ be the least index such that $q_{t}$ is the nonaccepting state in question. For each $j < t$, we repeat the same dimension categorization procedure as before. All remaining dimensions satisfy $m \geq 1 + d_{t}$ and $n \geq 1 + r_{t}$; for input words having these dimensions, the nonaccepting state $q_{t}$ is reached before the input head encounters a boundary, so $\mathcal{A}$ does not accept and we categorize that dimension into $R_{\mathcal{A}}$.

\item Suppose a nonaccepting state repeats in the sequence. Let $t$ be the least index such that $q_{s} = q_{t}$ for some index $s < t$; the sequence of states up to $q_{s}$ forms a prefix, while the sequence from $q_{s}$ to $q_{t}$ forms a cycle. Iterating this cycle once moves the input head of $\mathcal{A}$ downward by $D_{0} = d_{t} - d_{s}$ cells and rightward by $R_{0} = r_{t} - r_{s}$ cells; since $\mathcal{A}$ must move its input head on each transition, at least one of $D_{0}$ and $R_{0}$ is positive. For each $j < s$, we repeat the same dimension categorization procedure as before. For each $s \leq j < t$, consider the position of the input head after $k \geq 0$ iterations of the cycle. Immediately before $\mathcal{A}$ follows transition $j$, its input head has moved downward by $kD_{0} + d_{j}$ cells and rightward by $kR_{0} + r_{j}$ cells, and we record the input word dimensions $(m, n)$ for which transition $j$ would move the input head of $\mathcal{A}$ into the bottom or right boundary.
	\begin{itemize}
	\item If $X_{j} = D$, then this happens when the input word dimensions satisfy $m = 1 + kD_{0} + d_{j}$ and $n \geq 1 + kR_{0} + r_{j}$ for $k \geq 0$.
	
	\item If $X_{j} = R$, then this happens when the input word dimensions satisfy $m \geq 1 + kD_{0} + d_{j}$ and $n = 1 + kR_{0} + r_{j}$ for $k \geq 0$.
	\end{itemize}
Again, we categorize each case as accepting or rejecting based on the IBR-accepting behaviour of $\mathcal{A}$ upon reaching the boundary.
\end{itemize}
Finally, we take the accepted dimension set $D(\mathcal{A})$ to be the union of the accepting cases in $A_{\mathcal{A}}$, and we also construct a `rejected dimension' set $X(\mathcal{A})$ by taking the union of the rejecting cases in $R_{\mathcal{A}}$. Since every sequence of states has length at most $|Q|$ before one of our three cases is encountered, our construction is polynomial in $|Q|$.

Now, observe that we can formulate each of our decision problems in terms of these two sets: universality asks whether $X(\mathcal{A}) = \emptyset$, equivalence asks whether $(D(\mathcal{A}) \cap X(\mathcal{B})) \cup (D(\mathcal{B}) \cap X(\mathcal{A})) = \emptyset$ for some other unary two-way deterministic two-dimensional automaton $\mathcal{B}$, inclusion asks whether $D(\mathcal{A}) \cap X(\mathcal{B}) = \emptyset$, and disjointness asks whether $D(\mathcal{A}) \cap D(\mathcal{B}) = \emptyset$.

Since the sets $A_{\mathcal{A}}$ and $R_{\mathcal{A}}$ on which each of these decision problem definitions is predicated are constructed according to a case-based system of linear equalities and inequalities in a fixed number of integer variables (e.g., $m$, $n$, $k$), and since the number of variables is bounded by a constant independent of $\mathcal{A}$, then we can model each of these decision problems as a type of constraint satisfaction problem. Since an integer linear program with a fixed number of variables can be solved in polynomial time~\cite{Lenstra1983IntegerProgrammingFixedVariables}, each of these decision problems is in \P.
\end{proof}
\end{theorem}


\subsection{Nondeterministic Decidability and Complexity}\label{subsec:2NFA2W}

In contrast to the deterministic case, and similar to the unary three-way nondeterministic model, the decidability status of the \TNFATWOWOS\ universality, equivalence, inclusion, and disjointness problems remained unresolved in the literature. Here, we prove that the unary variants of each of these problems are decidable for the model, and we additionally establish upper and lower complexity bounds for each problem.

Unfortunately, since we are dealing with the nondeterministic model here, we cannot use the same techniques as in Section~\ref{subsec:2DFA2W} to establish decidability or prove complexity bounds; we could still analyze the sequence of states visited by an automaton in repeatedly reading the same symbol, but this sequence is no longer guaranteed to be unique. Instead, we will analyze a different kind of sequence: since every computation of a two-way two-dimensional automaton consists of some combination of downward and rightward moves, we can study the sequences of moves over $\{D, R\}^{*}$ that arise from such a computation and use existing tools to decide each of our problems.

To begin, we recall a helpful result from the literature, which we present here without proof.

\begin{proposition}[Ginsburg and Spanier~\cite{GinsburgSpanier1966PresburgerFormulas}]\label{prop:semilinearpresburger}
A set is semilinear if and only if the set is definable in Presburger arithmetic.
\end{proposition}

With this result, we can prove the following technical lemma.

\begin{lemma}\label{lem:2NFA2Wpresburger}
Given a unary two-way nondeterministic two-dimensional automaton $\mathcal{A}$ over an alphabet $\Sigma = \{\texttt{a}\}$, one can obtain a Presburger arithmetic formula $\phi_{\mathcal{A}}(m, n)$ such that, for all $m, n \geq 1$, $\phi_{\mathcal{A}}(m, n)$ is true if and only if $\texttt{a}^{m \times n} \in L(\mathcal{A})$.

\begin{proof}
First, observe that since we are operating over a unary alphabet, checking whether $\texttt{a}^{m \times n} \in L(\mathcal{A})$ is equivalent to checking whether $(m, n) \in D(\mathcal{A})$.

If $q_{0} = q_{\text{accept}}$, then take $\phi_{\mathcal{A}}(m, n) = (m \geq 1 \wedge n \geq 1)$. Otherwise, every computation of $\mathcal{A}$ corresponds to a sequence of moves $v \in \{D, R\}^{*}$, so we may represent computations in terms of the Parikh map of this sequence of moves $\Psi(v) = (|v|_{D}, |v|_{R})$. If $\Psi(v) = (d, r)$, then the sequence of moves $v$ has caused the input head of $\mathcal{A}$ to move $d$ rows downward and $r$ columns rightward within its input word.

We now classify accepting computations of $\mathcal{A}$ into one of the following cases:
\begin{itemize}
\item Suppose $\mathcal{A}$ accepts within the input word. For each accepting transition $e$ of the form $(q_{\text{accept}}, X) \in \delta(p, \texttt{a})$, where $X \in \{D, R\}$, let $V^{\text{int}}_{e} \subseteq \{D, R\}^{*}$ denote the set of sequences of moves sending $\mathcal{A}$ from $q_{0}$ to $p$. Observe that $V^{\text{int}}_{e}$ is a regular language, so by Parikh's theorem~\cite{Parikh1966OnCFLs}, its Parikh map $\Psi(V^{\text{int}}_{e})$ is semilinear. Then, by Proposition~\ref{prop:semilinearpresburger}, there exists a Presburger arithmetic formula $\phi^{\text{int}}_{e}(d, r)$ that is true if and only if there exists some $v \in V^{\text{int}}_{e}$ such that $\Psi(v) = (d, r)$. The set of all dimensions $(m, n)$ where $\mathcal{A}$ accepts within the input word having dimension $m \times n$ is then given by the formula
\begin{equation*}
\phi^{\text{int}}(m, n) = \bigvee_{e} \exists d \ \exists r \ (\phi^{\text{int}}_{e}(d, r) \wedge m \geq d + 1 \wedge n \geq r + 1).
\end{equation*}

\item Suppose $\mathcal{A}$ accepts upon reaching the bottom boundary of the input word. The argument is similar to that used in the previous case, except we now consider the set of sequences of moves where the last transition moves the input head of $\mathcal{A}$ downward into the boundary. In this case, the set of all dimensions $(m, n)$ where $\mathcal{A}$ accepts at the bottom boundary of the input word having dimension $m \times n$ is given by the formula
\begin{equation*}
\phi^{\text{bot}}(m, n) = \bigvee_{e} \exists d \ \exists r \ (\phi^{\text{bot}}_{e}(d, r) \wedge m = d \wedge n \geq r + 1).
\end{equation*}

\item Suppose $\mathcal{A}$ accepts upon reaching the right boundary of the input word. The argument is similar to those used in the previous two cases, except we now consider the set of sequences of moves where the last transition moves the input head of $\mathcal{A}$ rightward into the boundary. In this case, the set of all dimensions $(m, n)$ where $\mathcal{A}$ accepts at the right boundary of the input word having dimension $m \times n$ is given by the formula
\begin{equation*}
\phi^{\text{rig}}(m, n) = \bigvee_{e} \exists d \ \exists r \ (\phi^{\text{rig}}_{e}(d, r) \wedge m \geq d + 1 \wedge n = r).
\end{equation*}
\end{itemize}
Finally, we take $\phi_{\mathcal{A}}(m, n) = \phi^{\text{int}}(m, n) \vee \phi^{\text{bot}}(m, n) \vee \phi^{\text{rig}}(m, n)$.

Now, suppose $\texttt{a}^{m \times n} \in L(\mathcal{A})$; equivalently, suppose $(m, n) \in D(\mathcal{A})$. Since $\mathcal{A}$ is a two-way two-dimensional automaton, and therefore IBR-accepting, the accepting computation must finish either within the input word, at the bottom boundary, or at the right boundary. In any case, the corresponding sequence of moves belongs to one of $V^{\text{int}}_{e}$, $V^{\text{bot}}_{e}$, or $V^{\text{rig}}_{e}$, and so $\phi_{\mathcal{A}}(m, n)$ is true.

Conversely, suppose $\phi_{\mathcal{A}}(m, n)$ is true. Then one of $\phi^{\text{int}}(m, n)$, $\phi^{\text{bot}}(m, n)$, or $\phi^{\text{rig}}(m, n)$ is true. If $\phi^{\text{int}}(m, n)$ is true, then there exists a sequence of moves $v$ corresponding to an accepting computation of $\mathcal{A}$, and the inequalities $m \geq d + 1$ and $n \geq r + 1$ ensure acceptance occurs before the input head reaches a boundary. The arguments are analogous for the cases where $\phi^{\text{bot}}(m, n)$ or $\phi^{\text{rig}}(m, n)$ are true. Thus, $(m, n) \in D(\mathcal{A})$ and $\texttt{a}^{m \times n} \in L(\mathcal{A})$.
\end{proof}
\end{lemma}

It is well known that Presburger arithmetic is decidable~\cite{Presburger1930Arithmetik}, and Oppen~\cite{Oppen1978UpperBoundPresburger} provided a decision procedure for Presburger arithmetic that belongs to the complexity class \ELEMENTARY. Using Lemma~\ref{lem:2NFA2Wpresburger} to reframe everything in the system of Presburger arithmetic, we can make use of this existing decision procedure.

\begin{theorem}\label{thm:unary2NFA2Wuniversality}
The universality, equivalence, inclusion, and disjointness problems for two-way nondeterministic two-dimensional automata over a unary alphabet are in \ELEMENTARY.

\begin{proof}
Let $\mathcal{A}$ and $\mathcal{B}$ be unary two-way nondeterministic two-dimensional automata. Use Lemma~\ref{lem:2NFA2Wpresburger} to obtain Presburger arithmetic formulas $\phi_{\mathcal{A}}(m, n)$ and $\phi_{\mathcal{B}}(m, n)$ corresponding to the accepted dimension sets $D(\mathcal{A})$ and $D(\mathcal{B})$, respectively. Then, we may express each of our decision problems as Presburger arithmetic formulas in the following way:
\begin{itemize}
\item Universality: $\forall m \ \forall n \ ((m \geq 1 \wedge n \geq 1) \Rightarrow \phi_{\mathcal{A}}(m, n))$.
\item Equivalence: $\forall m \ \forall n \ ((m \geq 1 \wedge n \geq 1) \Rightarrow (\phi_{\mathcal{A}}(m, n) \Leftrightarrow \phi_{\mathcal{B}}(m, n)))$.
\item Inclusion: $\forall m \ \forall n \ ((m \geq 1 \wedge n \geq 1) \Rightarrow (\phi_{\mathcal{A}}(m, n) \Rightarrow \phi_{\mathcal{B}}(m, n)))$.
\item Disjointness: $\forall m \ \forall n \ ((m \geq 1 \wedge n \geq 1) \Rightarrow \neg(\phi_{\mathcal{A}}(m, n) \wedge \phi_{\mathcal{B}}(m, n)))$. \qedhere
\end{itemize}
\end{proof}
\end{theorem}

Unfortunately, if we wish to improve this upper bound substantially, we would likely require a novel decision procedure that does not rely on Presburger arithmetic. As Fischer and Rabin~\cite{FischerRabin1974SuperExponentialPresburger} showed, there exists a constant $c > 0$ such that every decision procedure for Presburger arithmetic has a worst-case time complexity lower bound of $2^{2^{cn}}$, and so any decision procedure using Presburger arithmetic will therefore still have a multiply exponential time complexity.

On the other hand, we can obtain lower bounds on the complexity of each of these decision problems via straightforward reductions, though the gap between our upper and lower bounds is quite large.

\begin{theorem}\label{thm:unary2NFA2WcoNPhard}
The universality, equivalence, and inclusion problems for two-way nondeterministic two-dimensional automata over a unary alphabet are \co\NP-hard.

\begin{proof}
We prove this theorem via a reduction from the universality problem for unary nondeterministic finite automata. Let $\mathcal{M}$ be a unary nondeterministic finite automaton. We construct a unary two-way nondeterministic two-dimensional automaton $\mathcal{A}_{\mathcal{M}}$ with the property that $L(\mathcal{A}_{\mathcal{M}})$ is universal if and only if $L(\mathcal{M})$ is universal. Given an arbitrary input word $\texttt{a}^{m \times n}$, $\mathcal{A}_{\mathcal{M}}$ ignores the number of columns and simulates the computation of $\mathcal{M}$ on the input word $\texttt{a}^{m}$. Then, $\mathcal{A}_{\mathcal{M}}$ accepts $\texttt{a}^{m \times n}$ if and only if $\mathcal{M}$ accepts $\texttt{a}^{m}$. Since the universality problem for unary nondeterministic finite automata is \co\NP-complete and we can reduce from this problem to the universality problem for two-way nondeterministic two-dimensional automata, testing universality is \co\NP-hard.

Next, as in the proof of Corollary~\ref{cor:unary2DFA3Wequivalenceinclusion}, construct a unary two-way nondeterministic two-dimensional automaton $\mathcal{U}$ recognizing the language $\Sigma^{++}$. Since deciding whether $L(\mathcal{A})$ is universal is at least as hard as deciding whether $L(\mathcal{A}) = L(\mathcal{U})$ or $L(\mathcal{U}) \subseteq L(\mathcal{A})$, testing equivalence and inclusion is likewise \co\NP-hard.
\end{proof}
\end{theorem}

For the remaining lower bound, we recall the fact that \textsc{st-Connectivity}, the problem of determining whether a path exists between two vertices in a directed graph, is \NL-complete; this result is due to Jones~\cite{Jones1975SpaceBoundedReducibility}.

\begin{theorem}\label{thm:unary2NFA2WNLhard}
The disjointness problem for two-way nondeterministic two-dimensional automata over a unary alphabet is \NL-hard.

\begin{proof}
We prove this theorem via a reduction from \textsc{st-Nonconnectivity}; this problem is \co\NL-complete, and by the Immerman--Szelepcs\'{e}nyi theorem~\cite{Immerman1988NLClosedComplementation,Szelepcsenyi1988ForcedEnumeration}, it is also \NL-complete.

Let $G = (V, E)$ be a directed graph having distinguished vertices $s$ and $t$. We construct a unary two-way nondeterministic two-dimensional automaton $\mathcal{A}_{G}$ in the following way: each state $q_{v}$ of $\mathcal{A}_{G}$ corresponds to a vertex $v \in V$, the initial state corresponds to vertex $s$, the accepting state corresponds to vertex $t$, and for each edge $(u, v) \in E$, there exists a transition $(q_{v}, D) \in \delta(q_{u}, \texttt{a})$.

If $t$ is reachable from $s$ via a path having length $\ell$, then $\mathcal{A}_{G}$ accepts a unary input word having dimension $(\ell + 1) \times 1$. If $t$ is not reachable, however, then $L(\mathcal{A}_{G}) = \emptyset$. If we take $\mathcal{U}$ to be the unary two-way nondeterministic two-dimensional automaton constructed in the proof of Theorem~\ref{thm:unary2NFA2WcoNPhard}, then $L(\mathcal{A}_{G}) \cap L(\mathcal{U}) = \emptyset$ if and only if $t$ is not reachable from $s$ in $G$.

Since \textsc{st-Nonconnectivity} is \NL-complete and we can construct $\mathcal{A}_{G}$ in space logarithmic in $|G|$, testing disjointness is \NL-hard.
\end{proof}
\end{theorem}


\section{Conclusion}\label{sec:conclusion}

In this paper, we studied various decision problems for unary two-dimensional automata and established complexity bounds for these problems. We showed that the \TDFATWOS\ universality, equivalence, and inclusion problems were \co\NP-hard; the \TDFATWOWOS\ universality, equivalence, inclusion, and disjointness problems were in \P; the \TNFATWOWOS\ universality, equivalence, and inclusion problems were \coNP-hard and in \ELEMENTARY; and the \TNFATWOWOS\ disjointness problem was \NL-hard and in \ELEMENTARY. We also established that a bounded variant of the \TNFATWOS\ universality problem was decidable, and proved that it was \co\NP-complete.

There still remain certain questions that could guide future research in this area and allow us to better understand the capability of the unary two-dimensional automaton model. For instance, what is an upper bound on the complexity of the \TDFATWOS\ universality, equivalence, and inclusion problems? In addition, now that the \TNFATWOWOS\ universality, equivalence, inclusion, and disjointness problems are known to be decidable, it would be worth studying alternative decision procedures for these problems that do not rely on Presburger arithmetic; for such a conceptually simple model of computation, it would be somewhat surprising if the only known procedure to decide common problems for this model is in \ELEMENTARY. Lastly, there remain many open questions surrounding \TNFATWOS\ decision problems: is the universality problem decidable in the general case, and are each of the equivalence, inclusion, and disjointness problems decidable as well?


\section*{Acknowledgements}\label{sec:acknowledgements}

This research was supported by the Natural Sciences and Engineering Research Council of Canada (NSERC) Discovery Grant RGPIN-2024-04799.


\bibliographystyle{plain}
\bibliography{./References}


\end{document}